\documentclass[11pt,a4paper]{article}
\pdfoutput=1 %for arxiv
\usepackage[utf8]{inputenc}

\usepackage[margin=1in]{geometry}
\usepackage{amsfonts}
\usepackage{amsmath}
\usepackage{amssymb}
\usepackage{amsthm}
\usepackage{float}
\usepackage{algorithm}
\usepackage{algorithmic}
\usepackage[font=small,labelfont=bf]{caption} %customise text style on fig captions
\allowdisplaybreaks

\usepackage{xspace}

\usepackage{bold-extra} %adds bold+smallcaps (for standard fonts)

\usepackage{hyperref}
\usepackage[svgnames]{xcolor}
\hypersetup{colorlinks={true},urlcolor={blue},linkcolor={DarkBlue},citecolor=[named]{DarkGreen}}
\usepackage[authoryear,square]{natbib}
\usepackage{doi}

\usepackage{microtype}
\usepackage[capitalise,nameinlink,noabbrev]{cleveref}

\usepackage{tikz}
\usetikzlibrary{arrows.meta}

\usepackage{graphicx}

\usepackage{enumerate}

\newcommand{\NN}{\ensuremath{\mathbb{N}}}
\newcommand{\RR}{\ensuremath{\mathbb{R}}}
\DeclareMathOperator*{\argmax}{arg\,max}
\DeclareMathOperator*{\argmin}{arg\,min}

\newcommand{\WWL}{weakly well-layered\xspace}
\newcommand{\WWLcapital}{Weakly Well-Layered\xspace}

\newcommand{\PLS}{\textup{\textsf{PLS}}\xspace}
\newcommand{\TFNP}{\textup{\textsf{TFNP}}\xspace}
\newcommand{\NP}{\textup{\textsf{NP}}\xspace}

\newcommand{\EFX}{\textup{\textsc{Identical-EFX}}\xspace}
\newcommand{\KNESER}{\textup{\textsc{Kneser}}\xspace}
\newcommand{\FLIP}{\textup{\textsc{Flip}}\xspace}
\newcommand{\poly}{\ensuremath{\mathrm{poly}}}

\theoremstyle{definition}
\newtheorem{definition}{Definition}[section]

\newtheorem{example}{Example}

\theoremstyle{plain}
\newtheorem{theorem}{Theorem}[section]

\newtheorem{lemma}{Lemma}
\newtheorem{corollary}{Corollary}

\title{The Frontier of Intractability for EFX with Two Agents}

\author{
\begin{tabular}{c c c}
& \\ \textbf{Paul W. Goldberg} & \textbf{Kasper H\o gh} & \textbf{Alexandros Hollender}\\
\small{University of Oxford} & \small{Aarhus University} & \small{EPFL}\\
\href{mailto:paul.goldberg@cs.ox.ac.uk}{\small{\texttt{paul.goldberg@cs.ox.ac.uk}}} & \href{mailto:kh@cs.au.dk}{\small{\texttt{kh@cs.au.dk}}} & \href{mailto:alexandros.hollender@epfl.ch}{\small{\texttt{alexandros.hollender@epfl.ch}}}
\end{tabular}
}

\date{}

\begin{document}

\maketitle

\begin{abstract}
We consider the problem of sharing a set of indivisible goods among agents in a fair manner, namely such that the allocation is envy-free up to any good (EFX). We focus on the problem of computing an EFX allocation in the two-agent case and characterize the computational complexity of the problem for most well-known valuation classes. We present a simple greedy algorithm that solves the problem when the agent valuations are \WWL, a class which contains gross substitutes and budget-additive valuations. For the next largest valuation class we prove a negative result: the problem is \PLS-complete for submodular valuations. All of our results also hold for the setting where there are many agents with identical valuations.
\end{abstract}

\section{Introduction}

The field of fair division studies the following fundamental question: given a set of resources, how should we divide them among a set of agents (who have subjective preferences over those resources) in a fair way? This question arises naturally in many settings, such as divorce settlement, division of inheritance, or dissolution of a business partnership, to name just a few. Although the motivation for studying this question is perhaps almost as old as humanity itself, the first mathematical investigation of the question dates back to the work of Banach, Knaster and Steinhaus \citep{Steinhaus1948-fair-division,Steinhaus1949-pragmatique}.

Of course, in order to study fair division problems, one has to define what exactly is meant by a \emph{fair} division. Different fairness notions have been proposed to formalize this. Banach, Knaster and Steinhaus considered a notion which is known today as \emph{proportionality}: every agent believes that it obtained at least a fraction $1/n$ of the total value available, where $n$ is the number of agents. A generally\footnote{As long as agents' valuations are subadditive, every envy-free division also satisfies proportionality.} stronger notion, and one which seems more adapted to the motivating examples we mentioned above, is that of \emph{envy-freeness} \citep{gamow1958puzzle,Foley1966-thesis-EF,Varian1974-EF}. A division of the resources is said to be envy-free, if no agent is envious, i.e., no agent values the bundle of resources obtained by some other agent strictly more than what it obtained itself.

As our motivating examples already suggest, the case with few agents -- in fact, even just with two agents -- is very relevant in practice. When the resources are divisible, such as for example money, water, oil, or time, the fair division problem with two agents admits a very simple and elegant solution: the cut-and-choose algorithm, which already appears in the Book of Genesis. As its name suggests, in the cut-and-choose algorithm one agent cuts the resources in half (according to its own valuation), and the other agent chooses its preferred piece, leaving the other piece to the first agent. It is easy to check that this guarantees envy-freeness, among other things. The case of divisible resources, which is usually called \emph{cake cutting}, has been extensively studied for more than two agents. One of the main objectives in that line of research can be summarized as follows: come up with approaches that achieve similar guarantees to cut-and-choose, but for more than two agents. This has been partially successful, and notable results include the proof of the existence of an envy-free allocation for any number of agents \citep{Stromquist80-existence,Woodall80-existence,Su99-rental-harmony}, as well as a finite, albeit very inefficient, protocol for computing one \citep{AzizM16-EF-n-agents}.

In many cases, however, assuming that the resources are divisible might be too strong an assumption. Indeed, some resources are inherently \emph{indivisible}, such as a house, a car, or a company. Sometimes these resources can be made divisible by sharing them over time, for example, one agent can use the car over the weekend and the other agent on weekdays. But, in general, and in particular when agents are not on friendly terms with each other, as one would expect to often be the case for divorce settlements, this is not really an option.

Indivisible resources make the problem of finding a fair division more challenging. First of all, in contrast to the divisible setting, envy-free allocations are no longer guaranteed to exist. Indeed, this is easy to see even with just two agents and a single (indivisible) good that both agents would like to have. No matter who is given the good, the other agent will envy them. In order to address this issue of non-existence of a solution, various relaxations of envy-freeness have been proposed and studied in the literature. The strongest such relaxation, namely the one which seems closest to perfect envy-freeness, is called \emph{envy-freeness up to any good} and is denoted by EFX \citep{CaragiannisKMPSW19-MNW,GourvesMT14-near-EF}.
An allocation is EFX if for all agents $i$ and $j$, agent $i$ does not envy agent $j$, after removal of \emph{any} single good from agent $j$'s bundle.
In other words, an allocation is \emph{not} EFX, if and only if there exist agents $i$ and $j$, and a good in $j$'s bundle, so that $i$ envies $j$'s bundle even after removal of that good.

For this relaxed notion of envy-freeness, it is possible to recover existence, at least in some cases. An EFX allocation is guaranteed to exist for two agents with any monotone valuations \citep{PlautR20-EFX}, and for three agents if we restrict the valuations to be additive \citep{ChaudhuryGM20-EFX-three}. It is currently unknown whether it always exists for four or more agents, even just for additive valuations.

Surprisingly, proving the existence of EFX allocations for two agents is non-trivial. In order to use the cut-and-choose approach, we need to be able to ``cut in half''. In the divisible setting, this is straightforward. But, in the indivisible setting, we need to ``cut in half in the EFX sense,'' i.e., divide the goods into two bundles such that the first agent is EFX with either bundle. In other words, we first need to show the existence of EFX allocations for two identical agents, namely two agents who share the same valuation function, which is not a trivial task.

\citet{PlautR20-EFX} provided a solution to this problem by introducing the \emph{leximin++ solution}. Given a monotone valuation function, they defined a total ordering over all allocations called the leximin++ ordering. They proved that for two identical agents, the leximin++ solution, namely the global maximum with respect to the leximin++ ordering, must be an EFX allocation. As mentioned above, using the cut-and-choose algorithm, this shows the existence of EFX allocations for two, possibly different, agents. Unfortunately, computing the leximin++ solution is computationally intractable\footnote{Computing the leximin++ solution is \NP-hard, even for two identical agents with additive valuations. This can be shown by a reduction from the \textsc{Partition} problem (see \citep[Footnote 7]{PlautR20-EFX} and note that their argument, which they use for leximin, also applies to leximin++).} and so, while this argument proves the existence of EFX allocations, it does not yield an efficient algorithm.

Nevertheless, for two agents with \emph{additive} valuations, \citet{PlautR20-EFX} provided a polynomial-time algorithm based on a modification of the Envy-Cycle elimination algorithm of \citet{LiptonMMS04-indivisible}. They also provided a lower bound for the problem in the more general class of submodular valuations, but not in terms of computational complexity (i.e., not in the standard Turing machine model). Namely, they proved that for two identical agents with submodular valuations computing an EFX allocation requires an exponential number of queries in the query complexity model.

Their work naturally raises the following two questions about the problem of computing an EFX allocation for two agents:
\begin{enumerate}
    \item What is the \emph{computational} complexity of the problem for submodular valuations?
    \item What is the computational complexity of the problem for well-known valuation classes lying between additive and submodular,\footnote{In particular, \citet[Section 7]{PlautR20-EFX} propose studying the complexity of fair division problems with respect to the hierarchy of complement-free valuations ($\textsl{additive}\subseteq \textsl{OXS}\subseteq \textsl{gross substitutes}\subseteq\textsl{submodular}\subseteq\textsl{XOS}\subseteq\textsl{subadditive}$) introduced by \citet{LehmannLN06-complement-free-hierarchy}.} such as gross substitutes, OXS, and budget-additive?
\end{enumerate}

Note that it does not make sense to study the query complexity for additive valuations, since a polynomial number of queries is sufficient to reconstruct the whole valuation functions (and the amount of computation then needed to determine a solution is not measured in the query complexity). However, it does make sense to study the computational complexity of the problem for submodular valuations, as well as other classes beyond additive. The query lower bound by Plaut and Roughgarden essentially says that many queries are needed in order to gather enough information about the submodular valuation function to be able to construct an EFX allocation. But it does not say anything about the \emph{computational} hardness of finding an EFX allocation. Their lower bound does not exclude the possibility of a polynomial-time algorithm for submodular valuations in the standard Turing machine model. Studying the problem in the computational complexity model allows us to investigate how hard it is to solve when the valuation functions are given in some succinct representation, e.g., as a few lines of code, or any other form that allows for efficient evaluation.

\paragraph{\bf Our contribution.}
We answer both of the aforementioned questions:
\begin{enumerate}
    \item For submodular valuations, we prove that the problem is \PLS-complete in the standard Turing machine model, even with two identical agents.
    \item We present a simple greedy algorithm that finds an EFX allocation in polynomial time for two agents with \emph{\WWL} valuations, a class of valuation functions that we define in this paper and which contains all well-known strict subclasses of submodular, such as gross substitutes (and thus also OXS) and budget-additive.\footnote{The class of \WWL valuations also contains the class of \emph{cancelable} valuations which have been recently studied in fair division \citep{BergerCFF22-almost-full-EFX,AkramiACGMM23-EFX-three-simpler,AmanatidisBLLR23-round-robin-beyond}.}
\end{enumerate}
Together, these two results resolve the computational complexity of the problem for all valuation classes in the standard complement-free hierarchy ($\textsl{additive}\subseteq \textsl{OXS}\subseteq \textsl{gross substitutes}\subseteq\textsl{submodular}\subseteq\textsl{XOS}\subseteq\textsl{subadditive}$) introduced by \citet{LehmannLN06-complement-free-hierarchy}. Furthermore, just like in the work of \citet{PlautR20-EFX}, our negative and positive results also hold for any number of \emph{identical} agents.

Regarding the \PLS-completeness result, the membership in \PLS is easy to show using the leximin++ ordering of \citet{PlautR20-EFX}. The \PLS-hardness is more challenging. The first step of our hardness reduction is essentially identical to the first step in the corresponding query lower bound of \citet{PlautR20-EFX}: a reduction from a local optimization problem on the Kneser graph to the problem of finding an EFX allocation. The second step of the reduction is our main technical contribution: we prove that finding a local optimum on a Kneser graph is \PLS-hard\footnote{We note that proving a tight computational complexity lower bound is more challenging than proving a query lower bound, because we have to reduce from problems with more structure. Indeed, the exponential query lower bound for the Kneser problem (and thus also for the EFX problem) can easily be obtained as a byproduct of our reduction.}, which might be of independent interest.

\paragraph{\bf Further related work.}
The existence and computation of EFX allocations has been studied in various different settings, such as for restricted versions of valuation classes \citep{AmanatidisBFHV21-EFX,BabaioffEF-dichotomous}, when some items can be discarded \citep{CaragiannisGH19-donating,ChaudhuryKMS21-charity,BergerCFF22-almost-full-EFX,ChaudhuryGMMM21-rainbow}, or when valuations are drawn randomly from a distribution \citep{ManurangsiS21-asymptotic-EFX}.

A weaker relaxation of envy-freeness is \emph{envy-freeness up to one good} (EF1) \citep{Budish11,LiptonMMS04-indivisible}. It can be computed efficiently for any number of agents with monotone valuations using the Envy-Cycle elimination algorithm \citep{LiptonMMS04-indivisible}. If one is also interested in economic efficiency, then it is possible to obtain an allocation that is both EF1 and Pareto-optimal in pseudopolynomial time for additive valuations \citep{BarmanKV18-fair-efficient}. For more details about fair division of indivisible items, we refer to the recent survey by \citet{AmanatidisABFLMVW23-survey}.

\paragraph{\bf Outline.} We begin with \cref{sec:prelims} where we formally define the problem and solution concept, as well as some standard valuation classes of interest. In \cref{sec:algo} we introduce \emph{\WWL} valuation functions, and present our simple greedy algorithm for computing EFX allocations. Finally, in \cref{sec:pls} we prove our main technical result, the \PLS-completeness for submodular valuations.

\section{Preliminaries}\label{sec:prelims}

We consider the problem of discrete fair division where an instance consists of a set of agents $N$, a set of goods $M$, and for every agent $i\in N$ a valuation function $v_i\colon 2^{M}\rightarrow\RR_{\geq 0}$ assigning values to bundles of goods. All valuation functions will be assumed to be \emph{monotone}, meaning that for any subsets $S\subseteq T\subseteq M$ it holds that $v(S)\leq v(T)$, and \emph{normalized}, i.e., $v(\emptyset) = 0$.

We now introduce the different types of valuation functions that are of interest to us. A valuation $v\colon 2^M\rightarrow\RR_{\geq 0}$ is \emph{additive} if $v(S)=\sum_{g\in S}v(\{g\})$ for every $S\subseteq M$. The hardness result we present in \cref{sec:pls} holds for \emph{submodular} valuations. These are valuations that satisfy the following diminishing returns condition that whenever $S\subseteq T$ and $x\notin T$ it holds that $v(S\cup\{x\})-v(S)\geq v(T\cup\{x\})-v(T)$. 

Next, for our results in the positive direction, we introduce the classes of \emph{gross substitutes} and \emph{budget-additive} valuations, both contained in the class of submodular valuations. 
Before defining gross substitutes valuations, we have to introduce some notation. For a price vector $p\in \RR^{m}$ on the set of goods, where $m=|M|$, the function $v_p$ is defined by $v_p(S)=v(S)-\sum_{g\in S}p_g$ for any subset $S\subseteq M$, and the demand set is $D(v,p) = \argmax_{S\subseteq M}v_p(S)$. 
A valuation $v$ is \emph{gross substitutes} if for any price vectors $p,p'\in\RR^m$ with $p\leq p'$ (meaning that $p_g\leq p_g'$ for all $g\in M$), it holds that if $S\in D(v,p)$, then there exists a demanded set $S'\in D(v,p')$ such that $\{g\in S\colon p_g = p_g'\}\subseteq S'.$
That is to say, if some good $g$ is demanded at prices $p$ and the prices of some \emph{other} goods increase, then $g$ will still be demanded. These valuations have various nice properties, for instance guaranteeing existence of Walrasian equilibria \citep{Gul1999}.
Lastly, a valuation $v$ is \emph{budget-additive} if it is of the form $v(S)=\min\{B,\sum_{g\in S}w_g\}$ for reals $B,w_1,\dots, w_m\geq 0$. \citep{LehmannLN06-complement-free-hierarchy} show that a budget-additive valuation need not satisfy the gross substitutes condition. See \cref{fig:hierarchy} for the relationship between the valuation classes.

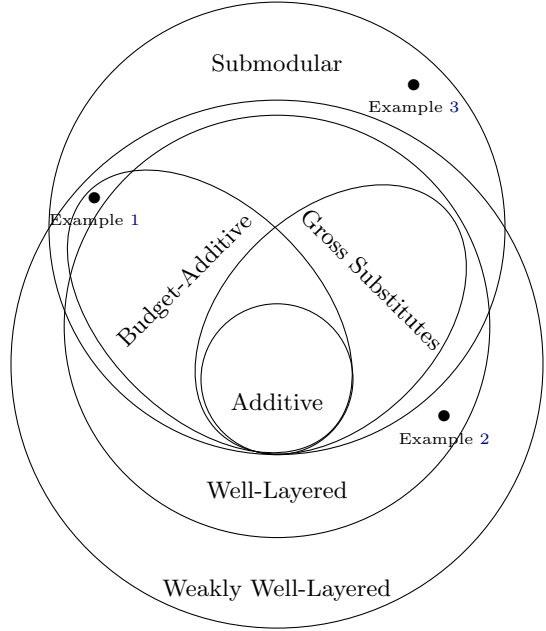
\begin{figure}
\centering
\begin{tikzpicture}
\draw (0,0.5) circle (3cm);
\node at (0,2.7) {\footnotesize{Submodular}};

\draw (0,-1.5) circle (1cm);
\node at (0,-1.8) {\footnotesize{Additive}};

\draw[rotate = 45] (0,-1) ellipse (2.1cm and 1.4cm);
\node[rotate=-45] at (1.25,-0.17) {\footnotesize{Gross Substitutes}};

\draw[rotate = -45] (-0.2,-1.05) ellipse (2.25cm and 1.4cm);
\node[rotate=45] at (-1.2,-0.2) {\footnotesize{Budget-Additive}};

\draw (0,-0.8) circle (2.8cm);
\node at (0,-3) {\footnotesize{Well-Layered}};

\draw (0,-1.3) circle (3.5cm);
\node at (0,-4.3) {\footnotesize{Weakly Well-Layered}};

\node at (-2.4,0.9){$\bullet$};
\node at (-2.4,0.6){\tiny Example \ref{ex:baval}};

\node at (2.2,-2.0){$\bullet$};
\node at (2.2,-2.3){\tiny Example \ref{ex:WL-not-submo}};

\node at (1.8,2.4){$\bullet$};
\node at (1.8,2.1){\tiny Example \ref{ex:algo-fail-submodular}};

\end{tikzpicture}
\caption{Inclusions of valuation classes} \label{fig:hierarchy}
\end{figure}

\paragraph{\bf Envy-freeness up to any good (EFX).}
The goal of fair division is to find an allocation of the goods to the agents (i.e., a partitioning $M = X_1\sqcup \cdots\sqcup X_n$) satisfying some notion of fairness. One might hope for an \emph{envy-free} division in which every agent prefers his own bundle over the bundle of any other agent, that is, $v_i(X_i)\geq v_i(X_j)$ for all $i,j\in N$. Such a division need not exist, however, as can be seen in the case where one has to divide one good among two agents, as already mentioned in the introduction. Therefore various weaker notions of fairness have been studied. In this paper we consider the notion of \emph{envy-freeness up to any good} (EFX) introduced by \citet{CaragiannisKMPSW19-MNW}, and before that by \citet{GourvesMT14-near-EF} under a different name. An allocation $(X_1,\dots, X_n)$ is said to be EFX if for any $i,j\in N$ and any $g\in X_j$ it holds that $v_i(X_i)\geq v_i(X_j\setminus \{g\})$.

\section{Polynomial-time Algorithm for \WWLcapital Valuations}\label{sec:algo}

In this section we present our positive result, namely the polynomial-time algorithm for computing an EFX allocation for two agents with \WWL valuations. To be more precise, our algorithm works for any number of agents that all share the same \WWL valuation function. As a result, using cut-and-choose it can then be used to solve the problem with two possibly \emph{different} agents. We begin with the definition of this new class of valuations, and then present the algorithm and prove its correctness.

\subsection{\WWLcapital Valuations}

We introduce a property of valuation functions and then situate this with respect to well-known classes of valuation functions in the next section.

\begin{definition}\label{def:WWL}
A valuation function $v\colon 2^M\rightarrow\RR_{\geq 0}$ is said to be \emph{weakly well-layered} if for any $M'\subseteq M$ the sets $S_0,S_1,S_2,\dots$ obtained by the greedy algorithm (that is, $S_0 = \emptyset$ and $S_i = S_{i-1}\cup \{x_i\}$ where $x_i \in\arg\max_{x\in M'\setminus S_{i-1}}v(S_{i-1}\cup \{x\})$  for $1\leq i\leq |M'|$) are optimal in the sense that $v(S_i ) = \max_{S\subseteq M'\colon |S|=i}v(S)$ for all $i$.
\end{definition}

We can reformulate this definition as follows: a valuation function $v$ is \WWL if and only if, for all $M' \subseteq M$ and all $i$, the optimization problem
\begin{equation}\label{eq:opt-WWL}
\begin{split}
\max \quad &v(S) \\
\text{ s.t. } \quad  
&S \subseteq M' \\
&|S| \leq i
\end{split}
\end{equation}
can be solved by using the natural greedy algorithm. Note that since we only consider monotone valuations, we can also use the condition $|S| = i$ instead of $|S| \leq i$.

The reformulation of the definition in terms of the optimization problem \eqref{eq:opt-WWL} is reminiscent of one of the alternative definitions of a matroid. Consider the optimization problem
\begin{equation}\label{eq:opt-matroid}
\begin{split}
\max \quad &v(S) \\
\text{ s.t. } \quad  
&S \in \mathcal{F}
\end{split}
\end{equation}
where $v\colon 2^M\rightarrow\RR_{\geq 0}$ is a valuation function and $\mathcal{F}$ is an independence system on $M$. Then, it is well-known that $\mathcal{F}$ is a matroid, if and only if, for all additive valuations $v$, the optimization problem \eqref{eq:opt-matroid} can be solved by the natural greedy algorithm \citep{Rado1957-independence,Gale1968-matroid,Edmonds1971-matroid-greedy}. In other words, the class of set systems (namely, matroids) is defined by fixing a class of valuations (namely, additive). The alternative definition of \WWL valuations given in \eqref{eq:opt-WWL} can be viewed as doing the opposite: the class of valuations (namely, \WWL) is defined by fixing a class of set systems (namely, all uniform matroids on subsets $M' \subseteq M$, or, more formally, $\mathcal{F} = \{S \subseteq M': |S| \leq i\}$ for all $M' \subseteq M$ and all $i$).

\subsection{Relationship to other valuation classes}

\paragraph{\bf Gross substitutes.}
We begin by showing that any gross substitutes valuation is \WWL. In particular, this also implies that OXS valuations, which are a special case of gross substitutes, are also \WWL. \citet{Leme2017} proved that gross substitutes valuation functions satisfy the stronger condition of being \emph{well-layered}, that is, for any $p\in\RR^m$ it holds that if $S_0,S_1,S_2,\dots$ is constructed greedily with respect to the valuation $v_p$, where $v_p(S):=v(S)-\sum_{g\in S}p_g$, then $S_i$ satisfies that $S_i \in\arg\max_{S\subseteq M\colon |S|=i}v_p(S)$. 

\begin{lemma}
If $v\colon 2^M\rightarrow\RR_{\geq 0}$ is well-layered, then it is also \WWL. In particular, gross substitutes valuations are \WWL.
\end{lemma}
\begin{proof}
Assume that $v\colon 2^M\rightarrow\RR_{\geq 0}$ is well-layered and let $M'\subseteq M$. Assume that the sequence $S_0,S_1,S_2,\dots$ is constructed via the greedy algorithm: that is $S_0 = \emptyset$ and $S_i = S_{i-1}\cup\{x_i\}$ where $x_i\in\argmax_{x\in M' \setminus S_{i-1}}v(S_{i-1}\cup\{x\})$ for $1 \leq i \leq |M'|$. We have to show that $v(S_i)=\max_{S\subseteq M'\colon |S|=i}v(S)$.

In order to exploit the assumption that $v$ is well-layered, we introduce a price vector $p\in\RR^m$ given by

\begin{align*}
    p_g = 
    \begin{cases}
    0 & g\in M' \\ 
    v(M)+1 & g\notin M'
    \end{cases}
\end{align*}

One sees that the sequence $S_0,S_1,S_2,\dots$ can occur via the greedy algorithm for the valuation $v_p$, because goods not in $M'$ cannot be chosen as their prices are too high. As $v$ is well-layered, we have that $v_p(S_i)=\max_{S\subseteq M\colon |S|=i}v_p(S)$. As $p_g=0$ for all $g\in M'$, this implies that $v(S_i)=\max_{S\subseteq M' \colon |S|=i}v(S)$. We conclude that $v$ is \WWL.
\end{proof}

\paragraph{\bf Closure properties and budget-additive valuations.}

We note that the class of \WWL valuations is closed under two natural operations.

\begin{lemma}
Let $v\colon 2^M\rightarrow\RR_{\geq 0}$ be weakly well-layered and let $f\colon \RR_{\geq 0}\rightarrow\RR_{\geq 0}$ strictly increasing. Then the composition $f\circ v\colon 2^M\rightarrow\RR_{\geq 0}$ is weakly well-layered.
\end{lemma}

\begin{proof}
Let $M'\subseteq M$ and assume that $S_0,S_1,S_2,\dots$ are constructed greedily, that is $S_0 = \emptyset$ and $S_i = S_{i-1}\cup \{x_i\}$ where $x_i\in\arg\max_{x\in M' \setminus S_{i-1}}f(v(S_{i-1}\cup\{x\}))$ for $1 \leq i \leq |M'|$. As $f$ is strictly increasing, we see that $x_i\in\arg\max_{x\in M'}f(v(S_{i-1}\cup\{x\}))$ if and only if $x_i\in\arg\max_{x\in M'}v(S_{i-1}\cup\{x\})$. 
Therefore $S_0,S_1,S_2,\dots$ could also arise via the greedy construction based on the valuation $v$. As $v$ is weakly well-layered, this implies that $v(S_i)=\max_{S\subseteq M'\colon |S|=i}v(S)$ for all $i$. As $f$ is increasing, this shows that $f(v(S_i))=\max_{S\subseteq M'\colon |S|=i}f(v(S))$ for all $i$. We conclude that $f\circ v$ is weakly well-layered.
\end{proof}

\begin{lemma}
Let $v\colon 2^M\rightarrow\RR_{\geq 0}$ be weakly well-layered and $B\geq 0$. Then the valuation $u\colon 2^M\rightarrow\RR_{\geq 0}$ given by $u(S)=\min(v(S),B)$ is \WWL.
\end{lemma}
\begin{proof}
Let $S_0,S_1,S_2,\dots $ be constructed greedily from the valuation $u.$ Suppose that $S_0,S_1,\dots, S_k$ have utility $<B$ and that $S_{k+1},S_{k+2},\dots$ have utility $B$. As $x\mapsto\min(x,B)$ is strictly increasing on $[0,B),$ the sets $S_0,S_1,\dots, S_k$ could have been constructed greedily from $v.$ As $v$ is \WWL, they are therefore optimal of their given size for $v$ and therefore also for $u.$ The sets $S_{k+1},\dots$ all have maximal utility $B$ and are therefore optimal of their given sizes. 
\end{proof}

As a corollary, since additive valuations are \WWL, it follows that the class of budget-additive valuations satisfies the \WWL property.

\begin{corollary}
Any budget-additive valuation is \WWL.
\end{corollary}

In contrast, it is known that budget-additive valuations are not necessarily gross substitutes, and, as the following example shows, not even well-layered.

\begin{example}\label{ex:baval}
Consider the budget-additive valuation on three goods $a,b,c$ with values $v_a=v_b=2$, $v_c=4$
and a budget of $B=4$. Let $p=(1,1,2)$ be a price vector. Under these prices, the greedy
algorithm would pick good $c$ as its first item. However, $\{a,b\}$ is the unique optimal bundle of size 2, and 
so the greedy algorithm would fail in this case. As a result, the valuation is not well-layered.
\end{example}

\paragraph{\bf Cancelable valuations.}

The class of \WWL valuations also contains the class of cancelable valuations recently defined by \citet{BergerCFF22-almost-full-EFX}, which contains budget-additive, unit-demand, and multiplicative valuations as special cases. A valuation function $v\colon 2^M\rightarrow\RR_{\geq 0}$ is said to be \emph{cancelable} if $v(S \cup \{x\}) > v(T \cup \{x\}) \implies v(S) > v(T)$ for any $S, T \subseteq M$ and $x \in M \setminus (S \cup T)$.

\begin{lemma}
Any cancelable valuation is \WWL.
\end{lemma}

\begin{proof}
Let $v$ be cancelable, $M' \subseteq M$, and let $S_0, S_1, S_2, \dots$ be obtained by the greedy algorithm on $v$ and $M'$ (see \cref{def:WWL}). We prove by induction that $v(S_i) = \max_{S \subseteq M': |S| = i} v(S)$ for all $i$. Clearly, this holds for $i=1$.

Now assume that the induction hypothesis holds for some $i \geq 1$ and consider $S_{i+1} = S_i \cup \{x_{i+1}\}$. If there existed $T \subseteq M'$ with $|T| = i+1$ such that $v(T) > v(S_{i+1})$, then, letting $y$ be any element in $T \setminus S_i$, we would obtain
$$v((T \setminus \{y\}) \cup \{y\}) = v(T) > v(S_{i+1}) = v(S_i \cup \{x_{i+1}\}) \geq v(S_i \cup \{y\})$$
where we used the fact that $x_{i+1}$ was added greedily to $S_i$. Since $v$ is cancelable, it follows that $v(T \setminus \{y\}) > v(S_i)$, which contradicts the induction hypothesis for $i$. As a result, the set $S_{i+1}$ must also be optimal.
\end{proof}

The results of this subsection are summarised in \cref{fig:hierarchy}. Note also that the classes of submodular valuations and \WWL valuations are incomparable. For an example of a valuation function that is submodular but not \WWL, see \cref{ex:algo-fail-submodular} in the next section. For the other direction, see the following example of a valuation that is well-layered (and thus \WWL), but not submodular.

\begin{example}\label{ex:WL-not-submo}
Consider the valuation function $v$ on two goods $a,b$ given by $v(\{a,b\})=1$ and $v(\emptyset)=v(\{a\})=v(\{b\})=0$. This valuation function is seen to be well-layered (and thus \WWL), because subsets of equal size have the same valuation. However, it is not submodular, because $v(\{a\}\cup\{b\})-v(\{a\})=1>0=v(\emptyset\cup\{b\})-v(\emptyset)$.
\end{example}

\subsection{The Greedy EFX Algorithm}

We now present a simple algorithm that computes an EFX allocation for many agents that all share the same \WWL valuation function $v$.

\begin{algorithm}[tb]
\caption{Greedy EFX}
\label{alg:algorithm}
\textbf{Input}: $N,M,v$\\
\textbf{Output}: EFX allocation
\begin{algorithmic}
\STATE Let $A_i = \emptyset$ for $i\in N$.
\STATE Let $R = M$.
\WHILE{$R\neq\emptyset$}
\STATE $i = \argmin_{j\in N}v(A_j)$
\STATE $g = \argmax_{x\in R}v(A_i\cup\{x\})$
\STATE $A_i = A_i\cup \{g\}$
\STATE $R = R\setminus\{g\}$
\ENDWHILE
\STATE \textbf{return} $(A_1,\dots, A_n)$
\end{algorithmic}
\end{algorithm}

\begin{theorem}
If the valuation function $v$ is \WWL, then the output of Algorithm 1 is EFX. In particular, by using the cut-and-choose protocol one may compute an EFX allocation for two agents with different valuations as long as one of these valuations is \WWL.
\end{theorem}

\begin{proof}
We show that the algorithm maintains a partial EFX allocation throughout. Initially the partial allocation is empty and so clearly EFX. Suppose that at the beginning of some round the current partial allocation $(X_1,\dots, X_n)$ is EFX and that some agent $i\in N$ receives a good $g$ in this round. We have to show that the new (partial) allocation $(X_1',\dots, X_n')$ is EFX, where $X_i ' = X_i\cup\{g\}$ and $X_j' = X_j$ for $j\neq i$. Clearly, the only thing we have to argue is that $v(X_i'\setminus\{g'\})\leq v(X_j')$ for all $j\in N$ and all $g'\in X_i'$. As $i$ received a good in the current round we have that $v(X_i)\leq v(X_j)=v(X_j')$. Therefore, it suffices to argue that $v(X_i'\setminus\{g'\})\leq v(X_i)$ for all $g'\in X_i '$. This last inequality follows from $v$ being \WWL by taking $M' = X_i'$. 
With this $M'$, the set $X_i$ could namely be produced by running the greedy algorithm. Therefore, $X_i$ is an optimal subset of $M'=X_i'$ of size $|X_i|=|X_i'|-1$, meaning that $v(X_i'\setminus\{g'\})\leq v(X_i)$ for all $g\in X_i'$.
\end{proof}

The algorithm can fail to provide an EFX allocation for submodular valuations that are not \WWL, as the following example shows.

\begin{example}\label{ex:algo-fail-submodular}
Consider an instance with two agents and four goods denoted $a,b,c,d$, where the valuation function $v$ is given by: $v(\{a\})=11, v(\{b\})=v(\{c\})=10, v(\{d\})=16, v(\{a,b\})=15, v(\{a,c\})=15, v(\{b,c\})=17, v(\{a,b,c\})=18$, and $v(S)=18$ for all sets $S$ that satisfy $d \in S$ and $|S| \geq 2$. It can be checked by direct computation that $v$ is indeed submodular. The greedy EFX algorithm yields: agent 1 gets good $d$, and then agent 2 gets goods $a,b,c$. This allocation is not EFX, because $v(\{d\}) < v(\{b,c\})$.
\end{example}

\section{\PLS-completeness for Submodular Valuations}\label{sec:pls}

\paragraph{\bf Total \NP search problems (\TFNP).} A total search problem is given by a relation $R \subseteq \{0,1\}^* \times \{0,1\}^*$ that satisfies: for all $x \in \{0,1\}^*$, there exists $y \in \{0,1\}^*$ such that $(x,y) \in R$. The relation $R$ is interpreted as the following computational problem: given $x \in \{0,1\}^*$, find some $y \in \{0,1\}^*$ such that $(x,y) \in R$. The class \TFNP \citep{MegiddoP91-TFNP} is defined as the set of all total search problems $R$ such that the relation $R$ is polynomial-time decidable (i.e., given some $x,y$ we can check in polynomial time whether $(x,y) \in R$) and polynomially balanced (i.e., there exists some polynomial $p$ such that $|y| \leq p(|x|)$ whenever $(x,y) \in R$).

Let $R$ and $S$ be two problems in \TFNP. We say that $R$ reduces to $S$ if there exist polynomial-time functions $f: \{0,1\}^* \to \{0,1\}^*$ and $g: \{0,1\}^* \times \{0,1\}^* \to \{0,1\}^*$ such that for all $x,y \in \{0,1\}^*$: if $(f(x),y) \in S$, then $(x,g(y,x)) \in R$. In other words, $f$ maps an instance of $R$ to an instance of $S$, and $g$ maps back any solution of the $S$-instance to a solution of the $R$-instance.

\paragraph{\bf Polynomial Local Search (\PLS).}

\citet{JohnsonPY88-PLS} introduced the class \PLS, a subclass of \TFNP, to capture the complexity of computing locally optimal solutions in settings where local improvements can be computed in polynomial time. In order to define the class \PLS, we proceed as follows: first, we define a set of basic \PLS problems, and then define the class \PLS as the set of all \TFNP problems that reduce to a basic \PLS problem.

A \emph{local search problem} $\Pi$ is defined as follows. For every instance\footnote{A more general definition would also include a polynomial-time recognizable set $D_{\Pi} \subseteq \{0,1\}^*$ of valid instances. The assumption that $D_{\Pi} = \{0,1\}^{*}$ is essentially without loss of generality. Indeed, for $I \notin D_{\Pi}$ we can define $F_I = \{0\}$, $c_I(0)=1$ and $N_I(0) = \{0\}$. Note that this does not change the complexity of the problem.} $I \in \{0,1\}^*$, there is a finite set $F_I \subseteq\{0,1\}^{*}$ of \emph{feasible solutions}, an objective function $c_I\colon F_I\rightarrow\NN$, and for every feasible solution $s \in F_I$ there is a neighborhood $N_I(s)\subseteq F_I.$ Given an instance $I$, one seeks a \emph{local optimum} $s^{*}\in F_I$ with respect to $c_I$ and $N_I$, meaning, in case of a maximization problem, that $c_I(s^{*})\geq c_I(s)$ for all neighbors $s\in N_I(s^{*})$.

\begin{definition}
A local search problem $\Pi$ is a basic \PLS problem if there exists some polynomial $p$ such that $F_I \subseteq \{0,1\}^{p(|I|)}$ for all instances $I$, and if there exist polynomial-time algorithms $A,B$ and $C$ such that:
\begin{enumerate}
\item Given an instance $I$, algorithm $A$ produces an initial feasible solution $s_0\in F_I$. 
\item Given an instance $I$ and a string $s \in \{0,1\}^{p(|I|)}$, algorithm $B$ determines whether $s$ is a feasible solution and, if so, computes the objective value $c_I(s)$. 
\item Given an instance $I$ and any feasible solution $s\in F_I$, the algorithm $C$ checks if $s$ is locally optimal and, if not, produces a feasible solution $s'\in N_I(s)$ that improves the objective value. 
\end{enumerate}
\end{definition}
Note that any basic \PLS problem lies in \TFNP.

\begin{definition}
The class \PLS is defined as the set of all \TFNP problems that reduce to a basic \PLS problem.
\end{definition}

A problem is \PLS-complete if it lies in \PLS and if every problem in \PLS reduces to it. \citet{JohnsonPY88-PLS} showed that the so-called \FLIP problem is \PLS-complete. We will define this problem later when we make use of it to prove our \PLS-hardness result.

\subsection{\PLS-membership}

\citet{PlautR20-EFX} prove the existence of an EFX allocation when all agents share the same monotone valuation, by introducing the leximin++ solution. In this section, we show how their existence proof can be translated into a proof of \PLS-membership for the following problem.

\begin{definition}[\EFX]
An instance $I=(N,M,C)$ of the \EFX search problem consists of a set of agents $N=[n]$, a set of goods $M=[m]$, and a boolean circuit $C$ with $m$ input gates. The circuit $C$ defines a valuation function $v\colon 2^M\rightarrow\NN$ which is the common valuation of all the agents. A solution is one of the following: 
\begin{enumerate}
    \item An allocation $(X_1,\dots, X_n)$ that is EFX.
    \item A pair of bundles $S\subseteq T$ that violate monotonicity, that is, $v(S)>v(T)$.
\end{enumerate}
\end{definition} 

The reason for allowing the violation-of-monotonicity solutions is that the circuit $C$ is not guaranteed to define a monotone valuation, and in this case an EFX allocation is not guaranteed to exist. Importantly, we note that our \PLS-hardness result (presented in the next section) does not rely on violation solutions. In other words, even the version of the problem where we are promised that the valuation function is monotone remains \PLS-hard.

\begin{theorem}
The \EFX problem lies in \PLS.
\end{theorem}

The problem of computing an EFX allocation for two non-identical agents with valuations $v_1$ and $v_2$ is reducible to the problem of computing an EFX allocation for two identical agents via the cut-and-choose protocol. As a result, we immediately also obtain the following:

\begin{corollary}
Computing an EFX allocation for two not necessarily identical agents is in \PLS. 
\end{corollary}

\begin{proof}
To show that the \EFX problem is in \PLS, we reduce it to a basic \PLS problem. An instance of this basic \PLS problem is just an instance of the \EFX problem, i.e, a tuple $I = (N,M,C)$. The set of feasible solutions $F_I$ is the set of all possible allocations of the goods in $M$ to the agents in $N$. As an initial feasible solution, we simply take the allocation where one agent receives all goods. It remains to specify the objective function $c_I$ and the neighborhood structure $N_I$, and then to argue that a local optimum corresponds to an EFX allocation.

\citet[Section~4]{PlautR20-EFX} introduce the \emph{leximin++} ordering on the set of allocations, and show that the maximum element with respect to that ordering must be an EFX allocation. In fact, a closer inspection of their proof reveals that even a \emph{local} maximum with respect to the leximin++ ordering must be an EFX allocation. As a result, we construct an objective function that implements the leximin++ ordering and then use the same arguments as \citet[Theorem~4.2]{PlautR20-EFX}.

For an allocation $(X_1,\dots, X_n)$, we let $O^X = (i_1,\dots, i_n)$ be an ordering of the agents according to increasing values of $v(X_i)$ (if multiple agents have bundles of equal utility, we break ties by ordering tied agents in terms of their agent number, i.e., if agents $i$ and $j$ are tied, and $i < j$, then agent $i$ will appear before agent $j$ in the ordering). The objective value is then defined as
\begin{align*}
c_I(X) =\, & |X_{i_n}|+v(X_{i_n})K\\
 +&|X_{i_{n-1}}|K^2+v(X_{i_{n-1}})K^3\\
 + & \dots \\
 + &|X_{i_{1}}|K^{2n-2}+v(X_{i_1})K^{2n-1}
\end{align*}
where $K$ is an upper bound on the size or utility of any bundle. We claim that if an allocation $X$ is not EFX, then one may construct an allocation $X'$ from $X$ by moving a single good from one bundle to another such that the objective strictly increases, $c_I(X')>c_I(X)$. Therefore, we will consider local maximization of this objective and we define the neighborhood of $X$ to be $N_I(X) = \{X'\in F_I\colon \exists i,j\in N, \exists g\in X_j\mbox{ s.t. }X'_i = X_i\cup\{g\},\; X'_j = X_j\setminus\{g\},\; X'_k=X_k\mbox{ for }k\neq i,j \}$. We note that the cardinality of $N_I(X)$ is polynomial in $n$ and $m$, so the algorithm for finding an improving neighbor if one exists may simply compute the objective value for every allocation in the neighborhood. Thus, this local maximization problem is indeed a basic \PLS problem.

Finally, we have to show that any local maximum $X\in F_I$ yields a solution to the \EFX problem, i.e., $X$ is an EFX allocation or $X$ yields a violation of monotonicity. We say that an allocation $X$ yields a violation of monotonicity, if there exist $i\in N$ and $g\in M$ such that $v(X_i\setminus \{g\})>v(X_i)$ or $v(X_i\cup \{g\})<v(X_i)$. We note that if $X$ yields a violation of monotonicity, then the violation can be found in polynomial time.

Consider an allocation $X\in F_I$ that is not EFX and that does not yield a violation of monotonicity. We will show that $X$ cannot be a local maximum, which then implies the desired statement by contrapositive.
Since $X$ is not EFX, we may find $i,j\in N$ and $g\in X_j$ such that $v(X_i)<v(X_j\setminus\{g\})$. Without loss of generality, we may assume that $i =\arg\min_{k\in N}v(X_k)$, and if more than one agent attains this minimum, then we take the $i$ that appears last among those tied agents in $O^X$ according to the tie-breaking. Now define an allocation $X'$ by
\begin{align*}
& X'_i = X_i\cup\{g\}\\
& X'_j = X_j\setminus\{g\}\\
& X'_k=X_k\mbox{ for }k\neq i,j
\end{align*}
and note that $X'\in N_I(X)$. We claim that $c_I(X')>c_I(X)$, meaning that $X$ is not a local maximum.

In order to prove this, we first show that the orderings $O^X$ and $O^{X'}$ agree in their first $\ell$ positions, where $\ell \in \{0,1, \dots, n-1\}$ is the index such that $O^X_{\ell +1} = i$, i.e., agent $i$ appears in position $\ell+1$ in $O^X$. Let $S$ denote the set of agents that appear in $O^X$ before agent $i$, i.e., the first $\ell$ agents appearing in $O^X$. Note that $S$ consists of all the agents that have utility $v(X_i)$ in allocation $X$, excluding $i$. First, observe that $j\notin S$, because $v(X_j)\geq v(X_j\setminus\{g\})>v(X_i)$ as $X$ does not yield a violation of monotonicity. Therefore, we find that the bundles of the agents in $S$ are not changed from allocation $X$ to $X'$, and, in particular, these agents still have utility $v(X_i)$ in allocation $X'$. Furthermore, in allocation $X'$, all other agents (except possibly $i$) have strictly larger utility than $S$-agents, namely $v(X'_j)=v(X_j\setminus\{g\})>v(X_i)$, and $v(X'_k)=v(X_k)>v(X_i)$ for $k\notin S\cup \{i,j\}$. Finally, $v(X'_i)=v(X_i\cup\{g\})\geq v(X_i)$ as $X$ does not yield a violation of monotonicity, and thus, in allocation $X'$, agent $i$ is either also tied with the agents in $S$, or it has strictly larger utility. In any case, by the tie-breaking, the first $\ell$ positions of $O^X$ and $O^{X'}$ are the same.

We now argue that $c_I(X')>c_I(X)$. Since $O^X$ and $O^{X'}$ agree in their first $\ell$ positions, and the bundles of those first $\ell$ agents have not changed, the $2\ell$ highest-order terms in $c_I(X)$ and $c_I(X')$ have identical coefficients. By definition, $O^{X}_{\ell+1}=i$. If $O^{X'}_{\ell+1}=i$, then we have that $c_I(X')>c_I(X)$, because $v(X'_i)=v(X_i\cup\{g\})\geq v(X_i)$ and $|X'_i|=|X_i\cup\{g\}|>|X_i|$, meaning that the coefficient in front of $K^{2n-(2\ell+1)}$ is at least as large in $c_I(X')$ as in $c_I(X)$ and the coefficient in front of $K^{2n-(2\ell+2)}$ is strictly larger.
If $O^{X'}_{\ell+1}=k\neq i$, then we have that $c_I(X')>c_I(X)$, because $v(X'_k)>v(X_i)$, implying that the coefficient in front of  $K^{2n-(2\ell+1)}$ is strictly larger in $c_I(X')$ than in $c_I(X)$. We conclude that $X$ is not a local maximum. 
Therefore, by contraposition, a local maximum is an EFX allocation or it yields a violation of monotonicity.
\end{proof}

\subsection{\PLS-hardness}

In this section we prove the following theorem.

\begin{theorem}\label{theorem:2agentsPLShard}
The problem of computing an EFX allocation for two identical agents with a submodular valuation function is \PLS-hard.
\end{theorem}

The reduction consists of two steps. First, following \citet{PlautR20-EFX}, we reduce the problem of local optimization on an odd Kneser graph to the problem of computing an EFX allocation for two agents sharing the same submodular valuation function. Then, in the second step, which is also our main technical contribution, we show that the \PLS-complete problem \FLIP reduces to local optimization on an odd Kneser graph.

\subsubsection{\texorpdfstring{$\KNESER \boldsymbol{\leq} \EFX$}{\KNESER < \EFX}}

For $k \in \NN$, the odd Kneser graph $K(2k+1,k)$ is defined as follows: the vertex set consists of all subsets of $[2k+1]$ of size $k$, and there is an edge between two vertices if the corresponding sets are disjoint. We identify the vertex set of $K(2k+1,k)$ with the set $\{x\in \{0,1\}^{2k+1}\colon ||x||_{1}=k\}$, where $||x||_{1}=\sum_{i=1}^{2k+1}x_i$ denotes the 1-norm. Note that there is an edge between $x$ and $x'$ if and only if $\langle x,x'\rangle =0$, where $\langle \cdot, \cdot \rangle$ denotes the inner product.

\begin{definition}[\KNESER]
The \KNESER problem of local optimization on an odd Kneser graph is defined as the following basic \PLS problem.
An instance of the \KNESER problem consists of a boolean circuit $C$ with $2k+1$ input nodes for some $k\in\NN$.
The set of feasible solutions is $F_C=\{x\in \{0,1\}^{2k+1}\colon ||x||_{1}=k\}$, and the neighborhood of some $x\in F_C$ is given by $N_C(x)=\{x'\in F_C \colon \langle x,x'\rangle = 0\}.$ The goal is to find a solution that is a local maximum with respect to the objective function $C(x) = \sum_{i=0}^{m-1}y_i\cdot 2^i$, where $y_0,\dots, y_{m-1}$ denote the output nodes of the circuit $C$. 
\end{definition}

\begin{lemma}
\KNESER reduces to \EFX with two identical submodular agents.
\end{lemma}

\begin{proof}
Our proof of this lemma closely follows the corresponding proof of \citet[Theorem~3.1]{PlautR20-EFX}, with some minor modifications due to the different computational model.
First, we describe the map $f$ taking instances $C$ of \KNESER to instances of \EFX.
We consider a valuation on subsets of $[2k+1]$ given by 
\begin{align*}
v(X) = \begin{cases}
2|X| & \mbox{if }|X|<k\\
2k-2^{-C(X)} & \mbox{if }|X|=k\\
2k & \mbox{if }|X|>k\\
\end{cases}
\end{align*}
Using the description of the circuit $C$, we may in polynomial time construct a boolean circuit computing $v$. This valuation may take non-integer values, but this can be fixed by scaling by a larger power of 2. Scaling will not change anything in the arguments below. We now define $f(C) = ([2],[2k+1],v)$. That is, the \KNESER instance $C$ is mapped to an \EFX instance with $2k+1$ goods and with two agents sharing the same valuation $v$.

We note that $2^{-C(X)}\in (0,1]$,
because $C$ takes values in the natural numbers. This ensures that the valuation $v$ is monotone, because $v(S)$ is seen to be non-decreasing in $|S|$.
Therefore, the only optimal solutions of $f(C)$ are EFX allocations $(X_1,X_2)$. Note by inspection of $v$ that if $(X_1,X_2)$ is EFX, then $|X_1|=k$ and $|X_2|=k+1$  (or $|X_1|=k+1$ and $|X_2|=k$). If we are in the first case then $X_1$ corresponds to a feasible solution of the \KNESER instance $C$. Also any neighbor of $X_1$ in the Kneser graph is of the form $X_2\setminus\{g\}$ for some $g\in X_2$. As $(X_1,X_2)$ is EFX we have that
\begin{align*}
2k-2^{-C(X_1)} & =v(X_1)\\
&\geq v(X_2\setminus\{g\}) =2k-2^{-C(X_2\setminus\{g\})}
\end{align*}
implying that $C(X_1)\geq C(X_2\setminus\{g\})$ for all $g \in X_2$. We conclude that $X_1$ is a local maximum for the instance of \KNESER given by the circuit $C$. Similarly, when $|X_2| = k$, $X_2$ will be a local maximum. As a result, we can define the polynomial-time map $g$ that maps solutions of the \EFX instance to solutions of the \KNESER-instance by
\begin{align*}
g((X_1,X_2),C) = \begin{cases}
X_1 & \mbox{if }|X_1|=k\\
X_2 & \mbox{otherwise}
\end{cases}
\end{align*}
By the discussion above it follows that if $(X_1,X_2)$ is a solution to the \EFX instance, then $g((X_1,X_2),C)$ is an optimal solution to the \KNESER-instance. Therefore, the pair $(f,g)$ constitutes a reduction from \KNESER to \EFX.

Finally, we show that $v$ is submodular. For any $X\subseteq [2k+1]$ and $x\notin X$ we have that
\begin{align*}
v(X\cup \{x\})-v(X)=\begin{cases}
2 & \mbox{if }|X|<k-1\\
2-2^{-C(X\cup\{x\})} & \mbox{if }|X|=k-1\\
2^{-C(X)} & \mbox{if }|X|=k\\
0 & \mbox{if }|X|>k
\end{cases}
\end{align*}
Using that $2^{-C(X)}\in (0,1]$, this shows that $v(X\cup\{x\})-v(X)$ is non-increasing in $|X|.$ Thus, if $Y\subseteq X$ and $x\notin X$, we have that $v(X\cup \{x\})-v(X)\leq v(Y\cup \{x\})-v(Y)$, meaning that $v$ is submodular. 
\end{proof}

\subsubsection{\texorpdfstring{$\FLIP \boldsymbol{\leq} \KNESER$}{\FLIP < \KNESER}}

\citet{JohnsonPY88-PLS} introduced the computational problem \FLIP and proved that it is \PLS-complete. We will now reduce from \FLIP to \KNESER to show that \KNESER, and thus \EFX, are \PLS-hard. In particular, this also establishes the \PLS-completeness of \KNESER, which might be of independent interest.

\begin{definition}[\FLIP]
The \FLIP problem is the following basic \PLS problem.
The instances of \FLIP are boolean circuits. For an instance $C$ with $n$ input nodes $x_0,\dots, x_{n-1}$ and $m$ output nodes $y_0,\dots,y_{m-1}$, the set of feasible solutions is all the possible inputs to the circuit: $F_C=\{0,1\}^n$. For any $x\in \{0,1\}^n$, the neighborhood is all the inputs that can be obtained from $x$ by flipping one bit: $N_C(x) = \{x'\in \{0,1\}^n\colon\Delta(x,x') = 1\}$ where $\Delta(\cdot,\cdot)$ denotes the Hamming distance. The goal is to find a solution that is locally minimal with respect to the objective function defined by $C(x)=\sum_{i=0}^{m-1}y_i\cdot 2^{i}$.
\end{definition}

\begin{lemma}
\FLIP reduces to \KNESER.
\end{lemma}
\begin{proof}
We construct a reduction from \FLIP to the minimization version of \KNESER. The minimization version of \KNESER is seen to be equivalent to its maximization version by negating the output nodes of the original circuit.
Let $C_F$ be an instance of \FLIP. Denote by $p=\poly (|C_F|)$ the length of the feasible solutions of $C_F$. 
The map of instances $f$ now takes $C_F$ to an instance $C_K$ of the \KNESER-problem whose feasible solutions are $F_K=\{x\in \{0,1\}^{2p+1}\colon ||x||_1 = p\}$. A typical feasible solution will be written as $s=uvb$ where $u,v\in\{0,1\}^p$ and $b\in\{0,1\}$. We will use the notation $\overline{u}$ to denote the bitwise negation of $u \in \{0,1\}^p$. The circuit $C_K$ is defined as follows:
\begin{enumerate}
\item $C_{K}(u\overline{u}0)=2\cdot C_{F}(u)$,
\item $C_{K}(uv1)=2\cdot \min(C_{F}(\overline{u}),C_F(v))+1$ if $\Delta(\overline{u},v)=1$,
\item $C_{K}(uvb)=M+\Delta(\overline{u},v)$ otherwise.
\end{enumerate}
Here $M$ denotes a number chosen to be sufficiently large so that it dominates any cost $2 \cdot C_F(w)$. Note that the circuit $C_K$ is well-defined and that it can be constructed in polynomial time given the circuit $C_F$. At a high level, the definition of the cost of a vertex of the third type is meant to ensure that for any such vertex $uvb$, there is a sequence of neighbors with decreasing costs that ends in a vertex of the form $u\overline{u}0$. The costs of the first and second vertex types are then meant to ensure that for a vertex $u\overline{u}0$, there is a sequence of neighbors with decreasing costs that ends in a vertex $w\overline{w}0$ where $w$ is an improving neighbor of $u$ in the original \FLIP-instance.

Below we show that the only local minima of $C_K$ are of the form $u\overline{u}0$ where $u$ is a local minimum for $C_F$. Therefore, upon defining the solution-mapping by $g(uvb)=u$ we have that $(f,g)$ is a reduction from \FLIP to \KNESER.

\textbf{No optimal solutions of type (3).}
If a feasible solution $s=uvb$ is of type (3), then we claim that it must have a neighbor of lower cost. First of all, note that since $s$ is not of type (1) or (2), and since $||s||_1 = p$, it follows that $\Delta(\overline{u},v) \geq 2$. Now, because $\Delta(\overline{u},v) \geq 2 > 0$ and $||uv||_1\leq p$, there must exist an $i$ such that $u_i=v_i=0$. Otherwise one would find that $||uv||_1 >p$, which contradicts $s$ being a feasible solution. Now, let $s' = u'v'b'$, where $u'=\overline{u}$, $b'=\overline{b}$, and $v'_j = \overline{v}_j$ for all $j \neq i$, but $v'_i = v_i = 0$.
We note that $||s'||_1=||\overline{s}||_1-1=(p+1)-1=p$, so $s'$ is a valid vertex in the Kneser graph. Further, we see that $s'$ is a neighbor of $s$, because $s_j ' s_j=0$ for all $j$. If $s'$ is not of type (3), then it has lower cost than $s$ by construction of $C_K$ and choice of $M$. Finally, if $s'$ is of type (3), then the observation that $\Delta(\overline{u'},v')<\Delta(\overline{u},v)$ again yields that $s'$ has lower cost than $s$.

\textbf{No optimal solutions of type (2).}
Suppose $s=uv1$ is of type (2). As $||s||_1=p$ and $\Delta(\overline{u},v)=1$, there is some $i$ with $v_i=0$ and $\overline{u}_i=1$, and $v_j=\overline{u}_j$ for $j\neq i$. This implies that $\sum_i u_i v_i = 0$, and so both $s'=\overline{u}u0$ and $s''=v\overline{v}0$ are neighbors of $s$. Furthermore, by construction of $C_K$, the cost of $s'$ or of $s''$ is strictly less than the cost of $s$.

\textbf{Optimal solutions.}
Consider a feasible solution of the form $u\overline{u}0$. If $u$ is not a local minimum for $C_F$, then let $w$ be an improving neighbor of $u$. As $\Delta(u,w)=1$, there are now two cases to consider. If $u_i=0$ and $w_i=1$ for some $i$, then $s' = \overline{w}u1$ is a type (2) neighbor of lower cost. If $u_i=1$ and $w_i=0$ for some $i$, then $s'=\overline{u}w1$ is a type (2) neighbor of lower cost. Therefore, if $u\overline{u}0$ is a local minimum for $C_K$, then $u$ is a local minimum for $C_F$.
\end{proof}

\begin{corollary}
Let $n\geq 2$ be an integer. Computing an EFX allocation for $n$ identical agents with a submodular valuation function is \PLS-hard.
\end{corollary}
\begin{proof}
By \cref{theorem:2agentsPLShard} it suffices to produce a reduction from the problem of computing an EFX allocation for two identical agents to the problem of computing an EFX allocation for $n$ identical agents. We sketch this reduction. Let $u\colon 2^M\rightarrow\RR$ denote the common submodular valuation function of the two agents. Construct an EFX-instance with $n$ agents by adding $n-2$ agents and $n-2$ goods, $M' = M\cup\{g_1,\dots, g_{n-2}\}$. Define the valuation function of the $n$ agents to be $u'=\overline{u}+v$ where $\overline{u}\colon 2^{M'}\rightarrow\RR$ is the extension of $u$ given by $\overline{u}(S)=u(S\cap M)$ and where $v\colon 2^{M'}\rightarrow\RR$ is additive given by $v(\{g_i\})=u(M)+1$ for $i=1,\dots, n-2$ and $v(\{g\})=0$ for $g\in M$. One may verify that $\overline{u}$ is submodular, and so that $u'$ is the sum of two submodular valuations and therefore itself submodular.

Let $(X_1,\dots, X_n)$ denote an EFX allocation of this instance. We claim that after permuting the bundles, we may assume that $X_{i+2} = \{g_i\}$ for $i=1,\dots, n-2$ and $X_1\cup X_2 = M$. At least one bundle, say $X_1$, receives no good from $\{g_1,\dots, g_{n-2}\},$ and so $u'(X_1)=u(X_1)\leq u(M).$ Now suppose some other bundle $X_i$ contains some good $g_j$. If $X_i$ contained another good $g$, then
\begin{align*}
    u'(X_i\setminus \{g\})\geq u'(\{g_j\})=u(M)+1>u'(X_1),
\end{align*}
contradicting $(X_1,\dots, X_n)$ being EFX. Hence, $X_i = \{g_j\}$, and the claim follows.
Now, one sees that $(X_1,X_2)$ is an EFX allocation of the original two-agent instance.
\end{proof}

\bigskip
\subsubsection*{Acknowledgements}

We thank all the reviewers of SAGT 2023 for their comments and suggestions that improved the presentation of the paper. In particular, we thank one reviewer for pointing out that \WWL valuations also generalize cancelable valuations.

P.\,W.\,Goldberg was supported by a JP Morgan faculty award. K.\,H\o gh was supported by the Independent Research Fund Denmark under grant no.~9040-00433B. Most of this work was done while he was visiting Oxford thanks to a STIBO IT Travel Grant. A.\,Hollender was supported by the Swiss State Secretariat for Education, Research and Innovation (SERI) under contract number MB22.00026.

\bibliographystyle{plainnat}
\bibliography{EFX_references}

\end{document}